\documentclass[11pt]{article}
\usepackage[english]{babel}
\usepackage{graphicx}
\usepackage{caption}
\usepackage{subcaption}
\usepackage{scrextend}
\usepackage{thmtools}
\usepackage{tikz}
\usepackage{caption}

	\usepackage{a4,geometry}

\usepackage{graphicx}
\usepackage{amsmath,amssymb,amsthm,mathtools}
\usepackage{paralist}
\usepackage{longtable}
\usepackage{colortbl}
\usepackage{hhline}
\usepackage{bm}
\usepackage{xspace}
\usepackage{url}
\usepackage{fullpage, prettyref}
\usepackage{boxedminipage}
\usepackage{wrapfig}
\usepackage{ifthen}
\usepackage{color}
\usepackage{xcolor}
\usepackage{framed}
\usepackage{algorithm}
\usepackage{algpseudocode}
\usepackage[pagebackref,letterpaper=true,colorlinks=true,pdfpagemode=none,urlcolor=blue,linkcolor=blue,citecolor=violet,pdfstartview=FitH]{hyperref}
\usepackage{fullpage}

\usepackage{thmtools}
\usepackage{thm-restate}

\newtheorem{theorem}{Theorem}[section]

\newtheorem{lemma}[theorem]{Lemma}
\newtheorem{claim}[theorem]{Claim}

\newtheorem{definition}[theorem]{Definition}

\newtheorem{fact}[theorem]{Fact}

\newtheorem{observation}[theorem]{Observation}
\newtheorem{construction}[theorem]{Construction}

\newcommand{\ignore}[1]{}

\DeclareMathOperator*{\argmin}{argmin} 

\newcommand{\cB}{\mathcal{B}}
\newcommand{\cC}{{\cal C}}
\newcommand{\cD}{\mathcal{D}}
\newcommand{\cE}{{\cal E}}
\newcommand{\cF}{\mathcal{F}}
\newcommand{\cG}{\mathcal{G}}

\newcommand{\cI}{{\cal I}}

\newcommand{\cM}{{\cal M}}

\newcommand{\cP}{\mathcal{P}}

\newcommand{\cT}{{\cal T}}

\newcommand{\N}{\mathbb N}

\newcommand{\eps}{\varepsilon}

\newcommand{\poly}{\mathrm{poly}}

\newcommand{\half}{1/2}

\newcommand{\bv}{\boldsymbol{v}}
\newcommand{\bw}{\boldsymbol{w}}

\newcommand{\bA}{\boldsymbol{A}}
\newcommand{\bC}{\boldsymbol{C}}

\newcommand{\Sec}[1]{\hyperref[sec:#1]{\S\ref*{sec:#1}}} 
\newcommand{\Eqn}[1]{\hyperref[eq:#1]{(\ref*{eq:#1})}} 
\newcommand{\Fig}[1]{\hyperref[fig:#1]{Fig.\,\ref*{fig:#1}}} 
\newcommand{\Tab}[1]{\hyperref[tab:#1]{Tab.\,\ref*{tab:#1}}} 
\newcommand{\Thm}[1]{\hyperref[thm:#1]{Theorem\,\ref*{thm:#1}}} 
\newcommand{\Fact}[1]{\hyperref[fact:#1]{Fact\,\ref*{fact:#1}}} 
\newcommand{\Lem}[1]{\hyperref[lem:#1]{Lemma\,\ref*{lem:#1}}} 
\newcommand{\Prop}[1]{\hyperref[prop:#1]{Prop.~\ref*{prop:#1}}} 
\newcommand{\Cor}[1]{\hyperref[cor:#1]{Corollary~\ref*{cor:#1}}} 
\newcommand{\Conj}[1]{\hyperref[conj:#1]{Conjecture~\ref*{conj:#1}}} 
\newcommand{\Def}[1]{\hyperref[def:#1]{Definition~\ref*{def:#1}}} 
\newcommand{\Alg}[1]{\hyperref[alg:#1]{Alg.~\ref*{alg:#1}}} 
\newcommand{\Clm}[1]{\hyperref[clm:#1]{Claim~\ref*{clm:#1}}} 
\newcommand{\Obs}[1]{\hyperref[obs:#1]{Observation~\ref*{obs:#1}}} 
\newcommand{\Rem}[1]{\hyperref[rem:#1]{Remark~\ref*{rem:#1}}} 
\newcommand{\Con}[1]{\hyperref[con:#1]{Construction~\ref*{con:#1}}} 
\newcommand{\Step}[1]{\hyperref[step:#1]{Step~\ref*{step:#1}}} 
\newcommand{\Assumption}[1]{\hyperref[assm:#1]{Assumption\,\ref*{assm:#1}}} 

\makeatletter

\newcommand{\ProbabilityRender}[2]{
  \@ifnextchar\bgroup%
  {\renderwithdist{#1}{#2}}
   {\singlervrender{#1}{#2}}
}
\newcommand{\singlervrender}[2]{%
   \ensuremath{\mathchoice
       {{#1}\left[ #2 \right]}
       {{#1}[ #2 ]}
       {{#1}[ #2 ]}
       {{#1}[ #2 ]}
   }
}
\newcommand{\renderwithdist}[3]{%
   \@ifnextchar\bgroup
   {\superfancyrender{#1}{#2}{#3}}
   {\ensuremath{\mathchoice
      {\underset{#2}{#1}\left[ #3 \right]}
      {{#1}_{#2}[ #3 ]}
      {{#1}_{#2}[ #3 ]}
      {{#1}_{#2}[ #3 ]}
     }
   }
}
\newcommand{\superfancyrender}[5]{
   \ensuremath{\mathchoice
      {\underset{#1}{{#1}}\left#4 #3 \right#5}
      {{#1}_{#2}#4 #3 #5}
      {{#1}_{#2}#4 #3 #5}
      {{#1}_{#2}#4 #3 #5}
   }
}

\newcommand{\deltaE}{\triangle E}
\newcommand{\comps}{\tt components}
\newcommand{\balsep}{{\tt BalSep}}
\newcommand{\decompose}{{\tt Decompose}}

\newcommand{\delsep}{{\tt Delsep}}

\newcommand{\E}{\mathbf{E}}

\newcommand{\row}[1]{\text{row}_{#1}}

\newcommand{\unlabr}{\texttt{UNLAB}_r}
\newcommand{\unlabu}{\texttt{UNLAB}_u}

\newcommand{\ct}[2]{\mathop{ct}_#2(#1)}
\newcommand{\allct}[1]{\mathop{ct}(#1)}
\newcommand{\dist}{\mathop{dist}}

\begin{document}
\renewcommand{\proofname}{\bfseries Proof:}

\title{The complexity of testing all properties of planar graphs, and the role of isomorphism}

\author{Sabyasachi Basu\thanks{Department of Computer Science, University of California, Santa Cruz. {\href{mailto:sbasu3@ucsc.edu}{sbasu3@ucsc.edu}}}
	\and Akash Kumar\thanks{School of Computer and Communication Sciences, École Polytechnique Fédérale de Lausanne. {\href{mailto:akash.kumar@epfl.ch}{akash.kumar@epfl.ch}} This project has received funding from the European Research Council (ERC) under the European Union’s Horizon 2020 research and innovation programme (grant agreement No 759471) }
	\and C. Seshadhri\thanks{Department of Computer Science, University of California, Santa Cruz. {\href{mailto:sesh@ucsc.edu}{sesh@ucsc.edu}}
	 \newline
	{SB and CS are supported by NSF DMS-2023495, CCF-1740850, CCF-1813165, CCF-1839317, CCF-1908384, CCF-1909790, and ARO Award W911NF1910294.}}
}
\maketitle

\begin{abstract} Consider property testing on bounded degree graphs and let $\varepsilon > 0$ denote the proximity parameter.  A remarkable theorem of Newman-Sohler (SICOMP 2013) asserts that \emph{all} properties of planar graphs (more generally hyperfinite) are testable with query complexity only depending on $\varepsilon$. Recent advances in testing minor-freeness have proven that all additive and monotone properties of planar graphs can be tested in $\mathrm{poly}(\varepsilon^{-1})$ queries. Some properties falling outside this class, such as Hamiltonicity, also have a similar complexity for planar graphs.  Motivated by these results, we ask: can all properties of planar graphs can be tested in $\mathrm{poly}(\varepsilon^{-1})$ queries?  Is there a uniform query complexity upper bound for all planar properties, and what is the ``hardest" such property to test?

	We discover a surprisingly clean and optimal answer. Any property of bounded degree planar graphs can be tested in $\exp(O(\varepsilon^{-2}))$ queries.  Moreover, there is a matching lower bound, up to constant factors in the exponent.  The natural property of testing isomorphism to a fixed graph requires $\exp(\Omega(\varepsilon^{-2}))$ queries, thereby showing that (up to polynomial dependencies) isomorphism to an explicit fixed graph is the hardest property of planar graphs.  The upper bound is a straightforward adapation of the Newman-Sohler analysis that tracks dependencies on $\varepsilon$ more carefully.  The main technical contribution is the lower bound construction, which is achieved by a special family of planar graphs that are all mutually far from each other. 

    We can also apply our techniques to get analogous results for bounded treewidth graphs. We prove that all properties of bounded treewidth graphs can be tested in $\exp(O(\varepsilon^{-1}\log \varepsilon^{-1}))$ queries. Moreover, testing isomorphism to a fixed forest requires $\exp(\Omega(\varepsilon^{-1}))$ queries.
\end{abstract}

\section{Introduction} \label{sec:intro}

Consider the setting of property testing for bounded degree graphs,
under the model of random access to a graph adjacency list, as 
introduced by Goldreich-Ron~\cite{GR02}. Let $G = (V,E)$ be a graph
where $V = [n]$ and the maximum degree is $d$.
We have random access to the list through \emph{neighbor queries}.
There is an oracle that, given $v \in V$ and $i \in [d]$,
returns the $i$th neighbor of $v$ (if no neighbor exists, it returns $\bot$).

For a property $\cP$ of graphs with degree bound $d$, the distance of $G$ to $\cP$
is the minimum number of edge additions/removals required to make $G$ have $\cP$,
divided by $dn$. We say that $G$ is $\eps$-far from $\cP$ if the distance to $\cP$
is more than $\eps$.  
A property tester for $\cP$ is a randomized procedure that takes as input (query access to) $G$ and a proximity parameter, $\eps > 0$.
If $G \in \cP$, the tester must accept with probability at least $2/3$. If $G$ is $\eps$-far from $\cP$,
the tester must reject with probability at least $2/3$. 
In our context, the property $\cP$ is called \emph{testable} if there
exists a property tester for $\cP$ whose query complexity is independent of $n$. 

One of the grand goals of property testing is to classify testable (graph) properties according to
the query complexity of testing them. Of special interest are \emph{efficiently testable properties},
whose testing complexity is $\poly(\eps^{-1})$. For bounded degree graphs,
there is little clarity on this issue. A recent survey by Goldreich states:
``Indeed, it is hard to find a common theme among [efficiently testable] properties...". \cite{Gol21}

Our work focuses on \emph{properties of planar graphs}. Even the problem of just testing planarity has received much attention, whose complexity has only recently been shown to be $\poly(\eps^{-1})$ \cite{BSS08,HKNO,EHNO11,YI:15,LR15,KSS:19}.
Newman-Sohler proved that every planar property (actually, every hyperfinite property) 
is testable, but they do not provide an explicit complexity bound depending on $\eps$ \cite{NS13}.
Kumar-Seshadhri-Stolman recently showed that every \emph{additive and monotone} 
planar property can tested in $\poly(\eps^{-1})$ queries~\cite{KSS:19}. Applying techniques from this 
result, Levi-Shoshan showed that planar Hamiltonicity can be tested efficiently~\cite{LS21}.
The latter property is neither additive nor monotone, so it is natural to ask
if all planar properties can be tested efficiently. If not, does there
exist a uniform complexity bound for all planar properties, and a candidate for 
the ``hardest planar property"? (We note that there is significant work on characterizing testable planar properties,
for the \emph{unbounded degree} case, which is qualitatively different~\cite{CMOS11, CS19}. Details are given in \Sec{related}.)

We discover that the answer to this question is the query complexity
bound $\exp(\Theta(\eps^{-2}))$. Up to constant factors in the exponent,
the hardest planar property is testing isomorphism to a fixed explicit graph.

Let us give some formalism, and discuss the connection to isomorphism. In our discussions, $n$ is basically fixed,
so we are considering a non-uniform setting. A \emph{planar property}
$\Pi$ is a set of unlabeled bounded degree planar graph with $n$ vertices. The input $G$
is a labeled graph and the tester is trying to property test if $G$ is isomorphic (or equal to,
if one ignores the labels) to any member of $\Pi$. Observe that singleton properties,
where $\Pi = \{H\}$, are equivalent to testing if $G$ is isomorphic to an fixed graph $H$.
With this premable, we can state our main results.

Our upper bound is a straightforward adaptation of arguments in Newman-Sohler~\cite{NS13}.

\begin{restatable}{theorem}{thmtestable} \label{thm:testable} Consider a planar property $\Pi$
of $n$-vertex, $d$ degree-bounded graphs. There is a property tester for $\Pi$ that makes
$\poly(d) \exp(O(\eps^{-2}))$ queries.
\end{restatable}

Our main technical result is a matching lower bound for an explicit singleton property.
One of the surprises (at least to the authors) is that testing isomorphism
to an arbitrary set of planar graphs is not harder, up to polynomial dependencies,
than testing isomorphism to a single fixed graph. A recent compendium of open
problems by Goldreich states the question of determining the complexity of testing
isomorphism (Open Problem 2.4 in~\cite{Gol21}). We note that our theorems resolve
this question for bounded-degree planar graphs.

\begin{restatable}{theorem}{thmplanarlb} \label{thm:main:planar}
For every sufficiently large $n$, there exists a bounded-degree planar graph $H$ on $n$ vertices 
such that property testing $\Pi = \{H\}$ requires $\exp(\Omega(\eps^{-2}))$ queries.
Equivalently, testing isomorphism to $H$ requires $\exp(\Omega(\eps^{-2}))$ queries.
%
\end{restatable}

{\bf Bounded treewidth classes:} In the context of property testing, many results
for planar graphs also hold for minor-free classes. It is natural to ask whether the
bound of $\exp(\Theta(\eps^{-2}))$ is the ``right" answer for any property 
of minor-free graphs. We discover that to \emph{not} be the case. For properties
of bounded treewidth graphs, the answer is between $\exp((\eps^{-1}))$ and $\exp((\eps^{-1})\log \eps^{-1})$.
The following theorems are also derived from the same methods for planar graphs,
but the constructions are substantially simpler. The lower bound is achieved
by a simple construction of forests. The main point of these results
is to show that the bounded of $\exp(\Theta(\eps^{-2}))$ achieved for planar
properties can be significantly beaten for non-trivial minor-closed families.

\begin{restatable}{theorem}{thmtwub} \label{thm:main:tw}
Consider a bounded treewidth property $\Pi$
of $n$-vertex, $d$ degree-bounded graphs. There is a property tester for $\Pi$ that makes
$\poly(d) \exp(O(\eps^{-1}\log \eps^{-1}))$ queries.
\end{restatable}

\begin{restatable}{theorem}{thmtwlb} \label{thm:main:tw:lb}
For every sufficiently large $n$, there exists a bounded-degree forest $H$ on $n$ vertices 
such that property testing $\Pi = \{H\}$ requires $\exp(\Omega(\eps^{-1}))$ queries.
\end{restatable}

\subsection{Main ideas} \label{sec:ideas}

One of the key components of property testers for planar graphs is the notion
of \emph{partition oracles for hyperfinite graphs}, as introduced by Hassidim-Kelner-Nguyen-Onak \cite{HKNO}. 
Bounded degree planar graphs (and more generally, minor-closed families) are hyperfinite,
meaning that one can remove a constant fraction of edges and obtain connected components of constant size.
Specifically, one can remove $\eps dn$ edges to get connected components of size $O(\eps^{-2})$.

Suppose $G$ is hyperfinite. A partition oracle is a local procedure that gives access to a hyperfinite decomposition to a graph $G$, without
any preprocessing. Given a query vertex $v$, the partition oracle outputs a connected subgraph $C(v)$ of size $O(\eps^{-2})$,
such that the union of components $C(v)$ forms a hyperfinite partition. Recent results give
a partition oracle that runs in time $\poly(\eps^{-1})$ per query~\cite{LR15, KSS:21}.

Newman-Sohler used partition oracles to prove that all hyperfinite (and hence planar) properties
are testable~\cite{NS13}. This is where our results begin. Stripping down the Newman-Sohler arguments to their core, 
we can essentially treat a planar $G$, up to $\eps dn$ edge changes, as a distribution over connected planar graphs
of size $O(\eps^{-2})$. 
The partition oracle allows us to sample efficiently from this distribution. (If the input graph $G$ is not hyperfinite,
the partition oracle can detect that efficiently, and the graph can be directly rejected.)
For convenience, let $\cD(G)$ denote this distribution.

Existing theorems in combinatorics
show that the number of planar graphs of size $O(\eps^{-2})$ is $\exp(O(\eps^{-2}))$~\cite{MR01}; hence,
this bounds the support size of $\cD(G)$. We can learn an approximation of $\cD(G)$
up to $\eps$ TV-distance with $\exp(O(\eps^{-2}))$ queries to the partition oracle.
It is not hard, but central to the argument, to show that if $\|\cD(G) - \cD(H)\|_1 \leq \eps$,
then $G$ is $\eps$-close to $H$.
For any property $\Pi$, one can simply check exhaustively if the learned approximation to $\cD(G)$ is close to $\cD(H)$ for any $H \in \Pi$.
(While this could be expensive in running time, it requires no further queries to $G$.)

The main challenge is in the lower bound. We need to find a property $\Pi$ such that 
testing $\Pi$ requires $\exp(\Omega(\eps^{-2}))$ samples from $\cD(G)$. We seem to need a converse to
the upper bound argument, showing if $\|\cD(G) - \cD(H)\|_1$ is large, then $G$ and $H$ are far from each other.
But this is false in general! The support of $\cD(G)$ is a set of graphs, which are mutable objects. Meaning,
we can modify $\cD(G)$ dramatically by only modifying a few edges of $G$. (As an extreme case, $\cD(G)$ and $\cD(H)$
could have disjoint supports on graphs that are close to each other.)

All our graphs (both input and hard instance) will be collections of connected components of size $O(\eps^{-2})$.
We begin by trying to construct a graph $H$ such that: for all $G$, $\|\cD(G) - \cD(H)\|_1 \geq \eps$
implies that $G$ is $\eps$-far from $H$. Moreover, determining if $\|\cD(G) - \cD(H)\|_1 \geq \eps$
should require $\exp(\Omega(\eps^{-2}))$ samples from $\cD(G)$. A candidate is suggested
by the problem of uniformity testing of distributions. Suppose we construct $H$ where $\cD(H)$ is a uniform
distribution on a collection $\cF$ of graphs (each member of which has size $O(\eps^{-2})$),
such that $|\cF| = \exp(\Omega(\eps^{-2}))$. By standard distribution testing arguments,
distinguishing the uniform distribution on $\cF$ from a uniform distribution on \emph{half}
of $\cF$ requires $\sqrt{|\cF|} = \exp(\Omega(\eps^{-2}))$ samples.
Suppose furthermore that all graphs in $\cF$
are $\eps$-far from each other and all balanced separators of $F \in \cF$ are at least of size $\eps |F|$.
We prove that any graph $G$ such that $\cD(G)$ is supported on half of $\cF$ must be $\eps$-far from $H$.
Moreover, we can construct candidate input graphs $G$, where any property tester can be simulated
by sampling from $\cD(G)$.
Putting all the arguments together, we get a bonafide lower bound: property testing $\Pi = \{H\}$
requires $\exp(\Omega(\eps^{-2}))$ queries.

What remains is the main technical construction. We need to design $\cF$, the ``suitable family"
of graphs for the lower bound. We achieve this family by taking the $\eps^{-1} \times \eps^{-1}$
grid and adding a collection of diagonal edges. We apply a cleaning procedure to get a large
family of graphs that are all far from each other. The separator size holds trivially,
since each graph in $\cF$ contains a large grid.

\subsection{Related Work} \label{sec:related}
Property testing on bounded-degree graphs is a vast topic and we direct the reader to Chapter 9 of Goldreich's book~\cite{Golbook} for an introduction to the subject. 

Arguably, the starting point for testing planarity is the seminal work of Benjamini-Schramm-Shapira \cite{BSS08}, who showed that all minor-closed properties are testable. This paper also
introduced the significance of hyperfiniteness to property testing. Hassidim-Kelner-Nguyen-Onak~\cite{HKNO} introduced the concept of partition oracles, a key tool in property testing for hyperfinite classes.
Improving the query time of partition oracles was addressed
by Edelman-Hassidim-Nguyen-Onak~\cite{EHNO11}, Levi-Ron~\cite{LR15} and Kumar-Seshadhri-Stolman~\cite{KSS:21}.

The quest for characterizing testable (bounded-degree) graph properties and finding properties testable in $\poly(\eps^{-1})$ is an important theme
in property testing. We note that Goldreich's recent survey explicitly calls these out as Open Problems
2.2 and 2.3~\cite{Gol21}. One of the main inspirations for our work is the result
of Newman-Sohler that proves that \emph{all} hyperfinite properties are testable~\cite{NS13}.
The recent work of~\cite{KSS:21} proves that all additive and monotone minor-closed properties
are efficiently testable, and Levi-Shoshan show that Hamiltonicity of minor-closed families
is also efficiently testable~\cite{LS21}. 

We note an important line of work of testing properties of planar graphs, in the \emph{unbounded
degree case}. This direction was pioneered by Czumaj-Monemizadeh-Onak-Sohler~\cite{CMOS11},
who showed the bipartitenesss is testable (independent of the size). A recent
result of Czumaj-Sohler prove that all testable planar properties in the unbounded degree setting
are related to subgraph freeness~\cite{CS19}.

\section{Proof of \Thm{testable} and \Thm{main:tw}} \label{sec:up}

We begin with some preliminaries: the notions of hyperfiniteness and partition oracles. Note
that all graphs have $n$ vertices and degree bound $d$.

\begin{definition} \label{def:hf} A graph $G$ is $(\eps, k)$ hyperfinite
if there exists a subset of $\eps dn$ edges whose removal results in connected components 
of size at most $k$.
\end{definition}

A classic result of Alon-Seymour-Thomas (Prop. 4.1 of~\cite{AST:94}) shows that all minor-closed
families are $(\eps, O(\eps^{-2}))$-hyperfinite. In \Lem{ASTtype}, we prove that all bounded treewidth
families are $(\eps, O(\eps^{-1}))$-hyperfinite. 
A partition oracle gives local access to a hyperfinite decomposition. We give the formal
definition below (adapted from Def. 1.1 of~\cite{KSS:21}).

\begin{definition} \label{def:oracle} 
A procedure $\bA$ is a partition oracle for a minor-closed family $\Pi$ of graphs 
if it satisfies the following properties. The deterministic procedure takes as input random access 
to $G = (V,E)$,
access to a random seed $r$, a proximity parameter $\eps > 0$, and a vertex $v$ of $G$.
(We will think of fixing $G, r, \eps$, so we use the notation $\bA_{G,r,\eps}$. All probabilities are with respect to $r$.)
The procedure $\bA_{G,r,\eps}(v)$ outputs a connected set of vertices, such that
the sets $\{\bA_{G,r,\eps}(v)\}$, over all $v$, form a partition of $V$. 

We say that the partition oracle outputs an $(\eps, k)$-hyperfinite decomposition if
the following properties hold. (i) For all $v$, $|\bA_{G,r,\eps}(v)| \leq k$
and (ii) with probability $> 2/3$ (over $r$),
the number of edges between the sets $\bA_{G,r,\eps}(v)$ is at most $\eps dn$.
\end{definition}

%
%
%
%
%
%
The main result of~\cite{KSS:21} that we use is the following.

\begin{theorem} [Rephrasing of Theorem 1.2 \cite{KSS:21}]\label{thm:kss:21}
	For any bounded degree graph in a minor-closed family, there is a partition oracle that outputs an $(\eps, O(\eps^{-2}))$-hyperfinite partition,
    and runs in time $O(\poly(d\eps^{-1}))$ per query.
\end{theorem}

For bounded treewidth classes, the partition oracle can output $(\eps, O(\eps^{-1}))$-hyperfinite partitions
using standard arguments  (\Clm{po:on:tw})

\subsection{Property testers through subgraph count vectors} \label{sec:count}

We use partition oracles to summarize a hyperfinite graph by an approximate count vector.
This technique was first used by Newman-Sohler~\cite{NS13}. We follow their analysis, but 
take care to keep track of various dependencies on $\eps$. This allows for getting
the optimal query complexity bound.

Consider any planar property $\Pi$ (technically, the following arguments hold for any minor-closed property).
For any planar $G$, let $\cP(G)$ denote the partition given by the partition oracle of \Thm{kss:21}.
(Note that the partition is a random variable.) 
%
%
Let $\cB_k$ be the set of unlabeled graphs in the family $\Pi$ with at most $k$ vertices.
We will always set $k =  O(\eps^{-2})$ (though the exact setting will change for bounded-treewidth properties).

For any graph $G'$ consisting of connected components of size at most $k$ and any $F \in \cB_k$,
let $\ct{G'}{F}$ be the number of occurrences of $F$ in $G'$. Let $\allct{G'}$ be the $|\cB_k|$-dimensional vector
of these counts. Note that $\|\allct{G'}\|_1 \in [n/k, n]$.

\begin{claim} \label{clm:dist} Consider two graphs $G'_1, G'_2$ consisting entirely of connected
	components of size at most $k$. If $\|\allct{G'_1} - \allct{G'_2}\|_1 \leq \gamma n$,
	then $\dist(G'_1, G'_2) \leq \gamma kd$.
\end{claim}

\begin{proof} For every $F \in \cB_k$, we will modify $G'_1$ and $G'_2$ to equalize
	the $\ct{G'_1}{F}$ and $\ct{G'_2}{F}$. We simply delete $|\ct{G'_1}{F} - \ct{G'_2}{F}|$
	instances from either $G'_1$ or $G'_2$ (whichever has the larger count). This operation
	deletes at most $|\ct{G'_1}{F} - \ct{G'_2}{F}|kd$ edges. In total,
	the number of edges deleted is at most $\|\allct{G'_1} - \allct{G'_2}\|_1 kd \leq \gamma nkd$.
\end{proof}

For the property $\Pi$ that we wish to test, construct the following set $\bC_\Pi$ of count vectors:
for every $F \in \Pi$ (with $n$ vertices and degree bound $d$) and every subgraph $F'$ that
is an $(\eps, k)$ partition of $F$, add $\allct{F'}$ to $\bC_\Pi$.

\begin{lemma} \label{lem:cpi} Suppose $\cP(G)$ is a valid $(\eps,k)$-partition of $G$. 
	If $G \in \Pi$, then $\allct{\cP(G)} \in \bC_\Pi$. If $G$ is $3\eps$-far from $\Pi$,
	then $\forall \bv \in \bC_\Pi$, $\|\allct{\cP(G)} - \bv\|_1 > \eps n/kd$.
\end{lemma}

\begin{proof} Suppose $G \in \Pi$. In the construction of $\bC_\Pi$ described above,
	we can select $F$ as $G$ and $F'$ as $\cP(G)$. Hence, $\allct{\cP(G)} \in \bC_\Pi$.

	Suppose $G$ is $3\eps$-far from $\Pi$. Consider any $\bv \in \bC_\Pi$,
	so $\bv = \allct{F'}$, for $F'$ being an $(\eps, k)$-partition of $F \in \Pi$.
	By triangle inequality, $\dist(\cP(G), F') \geq \dist(G,F) - \dist(G,\cP(G)) - \dist(F,F')$.
	By farness, $\dist(G,F) \geq 3\eps$. Because $\cP(G)$ and $F'$ are respective $(\eps,k)$-partitions,
	$\dist(G,\cP(G))$ and $\dist(F,F')$ are at most $\eps$. Hence, $\dist(\cP(G),F') \geq \eps$.
	By \Clm{dist}, $\|\allct{\cP(G)} - \bv\|_1 > \eps n/kd$.
\end{proof}

Now, we present the main algorithmic ingredient of our property tester.

\begin{claim} \label{clm:approx} 
	Given $G$ and a setting of 
	$\cP(G)$ that is a valid $(\eps,k)$-partition, using $O(\poly(d\eps^{-1})\delta^{-2} |\cB_k|^3 \log{|\cB_k|})$ 
	queries one can compute a vector $\bv$ such that 
	$\|\bv - \allct{\cP(G)}\|_1 < \delta n$ with probability at least $1-\frac{1}{|\cB_k|}$. 
\end{claim}

\begin{proof}  
	Fix $F \in \cB_k$. We show how to approximate $\ct{\cP(G)}{F}$. Pick uar vertex $s$,
	and using the partition oracle, determine the component of $\cP(G)$ containing $s$. If the component
	is isomorphic to $F$, declare success. The probability of success is exactly $|F|\ct{\cP(G)}{F}/n$.
	By Chernoff-Hoeffding, we can get an additive $\delta/|\cB_k|$ estimate with error 
	probability $< |\cB_k|^{-2}$ using $O(|\cB_k|^2 \delta^{-2} \log |\cB_k|)$ samples. 
	Thus, we get an additive $\delta n/(|F|\cdot |\cB_k|) \leq \delta n/|\cB_k|$
	estimate for $\ct{\cP(G)}{F}$. Applying for all $F$, we get our estimate vector $\bv$. 
	Note that the total approximation is 
	$\|\bv - \allct{\cP(G)}\|_1 < \delta n/|\cB_k| \times |\cB_k| = \delta n$.
	The error probability, by union bound, is at most $|\cB_k|^{-1}$. By the running time bound
    of the partition oracle, the total number of queries made is 
	$O(\poly(d\eps^{-1})\delta^{-2}|\cB_k|^3\log |\cB_k|)$.
\end{proof}

Now, we prove \Thm{testable}, repeated for convenience. To get the optimal bound
for planar properties, we need the right upper bound for $\cB_k$. Numerous
results in the past show that the answer is $\exp(O(k))$. We cite one specific reference.

\begin{theorem} \label{thm:mr} (Theorem 5 of~\cite{MR01}) A planar graph on $k$
vertices and $k'$ edges can be represented uniquely by $8k + 2k' + o(k+k')$ bits.

Hence, the number of planar graphs on $k$ vertices is $\exp(O(k))$.
\end{theorem}

\thmtestable*

\begin{proof} The tester has two phases. In the first phase, it uses \Thm{kss:21} to get
a hyperfinite partition. If it succeeds, then the tester proceeds to the second phase.
It uses \Clm{approx} to approximate $\allct{\cP(G)}$ and checks if it lies in $\bC_\Pi$. 

%
The tester repeats the following $O(1)$ times. For a random seed $R$,
it sets up the partition oracle with this seed, and then estimates the number of cut edges via random sampling.
This is done by sampling $\Theta(\eps^{-1})$  uar vertices $u$, picking a uar neighbor $v$ of $u$, and calling the partition oracle on $u$ and $v$. 
If the outputs are different components, then edge $(u,v)$ is cut. If more than an $\eps/4$-fraction of edges are cut,
the process is repeated with a new choice of $R$. Otherwise, $R$ is fixed, and the tester proceeds to the next phase.
If no $R$ is found, the tester rejects.

At this point, a suitable $\cP(G)$ has been discovered. We set $k = O(\eps^{-2})$, to
be the size bound of components, as promised by \Thm{kss:21}. Using \Clm{approx},
the tester computes an approximate
count vector $\bv$ of $\allct{\cP(G)}$, setting $\delta = \eps/4kd$. 
Note that $\bC_\Pi$ is a fixed set of vectors, independent of the input.
The tester
then determines if $\exists \bw \in \bC_\Pi$ such that $\|\bv - \bw\|_1 \leq \eps n/2kd$. If
such a vector $\bw$ exists, it accepts. Otherwise, it rejects.
The description of the property tester is complete, and we proceed to the analysis.

{\bf Query complexity bound:} Note that $k = O(\eps^{-2})$, $\delta = \eps/4kd$, 
and, by \Thm{mr}, $|\cB_k| = \exp(O(\eps^{-2}))$. Hence, by \Clm{approx}, the overall
query complexity is $\poly(d) \exp(O(\eps^{-2}))$.

{\bf Correctness analysis:} If $G$ is planar and hence $(\eps/8, O(\eps^{-2}))$-hyperfinite, with probability at least $2/3$ over $R$,
the partition oracle computes a hyperfinite decomposition. Hence, with probability at least $5/6$, the tester
proceeds to the second phase with $\cP(G)$ being a valid $(\eps, O(\eps^{-2}))$-hyperfinite partition.
Conversely, with probability at least $5/6$, if the tester finds a suitable $R$,
then $\cP(G)$ is a valid $(\eps, O(\eps^{-2}))$-hyperfinite partition.

Suppose $G \in \Pi$ and hence planar. As argued above, with probability $> 5/6$,
$\cP(G)$ is an $(\eps, O(\eps^{-2}))$-hyperfinite decomposition.
By \Clm{approx}, with 
probability $> 5/6$, $\|\bv - \allct{\cP(G)}\|_1 < \eps n/4kd$.
By the union bound, both conditions hold with probability at least $2/3$.
By \Lem{cpi}, $\allct{\cP(G)} \in \bC_\Pi$. Thus, there exists $\bw$
in $\bC_\Pi$ such that $\|\bv - \bw|_1 \leq \eps n/2kd$ and the tester accepts with probability at least $2/3$.

Suppose $G$ is $3\eps$-far from $\Pi$. By \Lem{cpi}, for all $\bw \in \bC_\Pi$,
$\|\allct{\cP(G)} - \bw\|_1 > \eps n/kd$. By the triangle inequality, 
$\|\bv - \bw\|_1 > \eps n/kd - \eps n/4kd > \eps n/2kd$.
Thus, the tester rejects with probability at least $2/3$.

%
%
%
%
\end{proof}

%
%
%
%
%

The proof of \Thm{main:tw} is nearly identical.

\thmtwub*
\begin{proof}
	We redo the proof of \Thm{testable}. \Clm{po:on:tw} shows that graphs with treewidth at most $\tau$ admit a 
	partition oracle with $k = O(\tau/\eps)$.
	For bounded treewidth graphs, we apply the trivial bound $|\cB_k| \leq 2^{O(k\log k)}$ (\Clm{count:tw}). 
    We apply these bounds in the above proof, and get a query complexity
    of $O(\poly(d) \exp(\eps^{-1}\log {\eps^{-1}}))$ queries.
\end{proof}

\section{Lower bounds through suitable families} \label{sec:suitable}

The key construct for the lower bounds is given in the following definition.
A graph class is monotone if it is closed under edge removals. A graph class
is additive if it closed under taking disjoint unions of graphs.
%

\begin{definition} \label{def:eps:suitable}
	Fix $\eps > 0$ and a monotone, additive graph class $\cC$. We call
	a family $\cF$ of distinct graphs (in $\cC$) \emph{$\eps$-suitable} if the following 
	conditions hold:
	\begin{asparaitem}
    \item All graphs in $\cF$ have the same number of vertices (denoted $t$).
	\item If $t > 2/\eps$, $\forall F,F' \in \cF$, $dist(F,F') \geq 0.02$.
	\item $\forall F \in \cF$, the size of any $(\frac{0.01}{d}, 1-\frac{0.01}{d})$ 
		balanced separator in $F$ is $\Omega(\eps t)$.
	\end{asparaitem}
\end{definition}

The main lemma used in our lower bound says that a large suitable family 
leads to a property testing lower bound.

\begin{restatable}{lemma}{abstractlb} \label{lem:main}
	Fix $\eps > 0$, and let $\cF$ denote an $\eps$-suitable family of a monotone, additive graph class $\cC$.
	Then, for all sufficiently large $n$, there exists a graph $H \in \cC$
    with the following property. Any property tester for the property $\{H\}$
    with proximity parameter $\eps$ requires $\Omega(\sqrt{|\cF|})$ queries.
\end{restatable}

The proof of \Lem{main} requires some tools, that we shall build
up in this section. Firstly, using the properties of a suitable family, we can construct two graphs
consisting entirely of components in $\cF$ that are far from each other.
These graphs form the core of the lower bound.

\begin{claim} \label{clm:main:item2}
Let $\cF$ denote an $\eps$-suitable family of a graph class $\cC$ (as defined in \Lem{main}).
Let $I = [|\cF|]$ denote a set indexing graphs in $\cF$. 
    
Let $H_1$ be the graph consisting of a disjoint union of all graphs in $|\cF|$.
For any subset $R \subseteq I$ with $|R| = |I|/2$, let $H_R$ be the disjoint
union of two copies of each graph in $\cF$ indexed by $R$.
Then the graphs $H$ and $H_R$ are $\Omega(\eps)$ far from each other.
\end{claim}

\begin{proof}
	Since $\cF$ is $\eps$-suitable, we know that all graphs $F \in \cF$ are defined on
	$t = t(\eps)$ vertices. We know for any $F, F' \in \cF$, since $|F| = |F'| = t$
	the number of edge edits needed to change $F$ to $F'$ is at least $\Omega(t)$.
    Since $R$ will be fixed for the remainder of the proof, it will be convenient to refer to $H_R$ as $H_2$
	for the rest of the argument.
	Let us denote the components of $H_1$ (resp $H_2$) as $\comps(H_1)$ (resp as $\comps(H_2)$). 
	Let $s = |\comps(\cF)|$. $H_2$ contains $s/2$ connected components in $\cF$ each of which occurs twice.
	Recall from \Def{eps:suitable}, $C,C' \in \comps(\cF)$ are $\Omega(1)$-far from each other. 
	Let $\cM_{\half}$ denote
	the set of \textit{missing} components of $\cF$ in $H_2$. And let $\cP_{\half}$ denote the
	set of components from $\cF$ that are \textit{present} in $H_2$ (without duplicates). Thus,
	$H_2$ contains two subgraphs isomorphic to $\cP_{\half}$ which we denote as $\cP^1_{\half}$
	and $\cP^2_{\half}$. For $C \in \comps(\cP_{\half})$, let $C^1 \in \comps(\cP^1_{\half})$
	denote the subgraph of $\cP^1_{\half}$ isomorphic to $C$. Similarly, define $C^2$. Let 
	$$\deltaE = \argmin_{\substack{\deltaE' \subseteq V(H_2) \times V(H_2)\\E(H_1) = E(H_2) \oplus \deltaE'}} |\deltaE'|$$ 
	denote the smallest set of edge modifications to $H_2$ that produce 
	a graph isomorphic to $H_1$. We now lower bound $|\deltaE|$. 
	For a component $C \in \cF$, let $E(C)$ denote the edge set of that component. 
	Fix a component $C$ such that $E(C) \cap \deltaE \neq \emptyset$. The following cases
	arise.

	\begin{itemize}
		\item \textbf{Case 1:} $\deltaE \cap E(C)$ is connected. Thus, the deletion edits
			do not disconnect $C$. In this case, after the insertion edits in $\deltaE$,
			we obtain a component in $\cM_{\half}$. By the definition of suitable families, if $t > 2/\eps$, it takes at least (follow the sentence). If $t \leq 2/\eps$, then $C$ obviously gets at least one edit, which is at least $\Omega(\eps t)$ edits. In both cases, the number of edits is $\Omega(\eps t)$.

		\item \textbf{Case 2:} $\deltaE \cap E(C)$ is not connected. Now, we get two subcases
			depending on size of the biggest component in $\deltaE \cap E(C)$.

			\begin{enumerate}
				\item   The biggest component in $C \setminus \deltaE$ has size at least 
					$(1 - 1/100d) \cdot t$. In this case, $\deltaE \oplus E(C)$ maps
					$C$ to some component in $\cM_{\half}$. Now we split into two cases.
                    Suppose $t > 2/\eps$. Since $C$ is a bounded degree
					graph, the number of insertion edits (by \Def{eps:suitable}) is at least
					$0.02t - (\frac{0.01}{d} \cdot d)t = 0.01t$. Suppose $t \leq 2/\eps$.
                    In any case $C$ must get at least one insertion edit to make it have size $t$.
                    So the number of edits is at least $\Omega(\eps t)$.

				\item The biggest component in $\deltaE \cap E(C)$ has size at most 
					$(1 - \frac{1}{100d}) \cdot t$. In this case, note that $\deltaE$ 
					removes a balanced separator of $C$. Thus, the number of deletion 
					edits made to $C$ is at least $\Omega(\eps t)$.
			\end{enumerate}

	\end{itemize}

	In all of the above cases we see that for any component $C \in \comps(H_2)$ with 
	$\deltaE \cap E(C) \neq \emptyset$, the number of edits made to $C$ is at least
	$\Omega(\eps t)$. Finally, note that 
	$$|\{C \in \comps(H_2) \colon E(C) \cap \deltaE \neq \emptyset\}| \geq \frac{|\comps(H_2)|}{2}.$$
	And therefore, the distance between graphs $H_1$ and $H_2$ is at least $\Omega(\eps)$ as desired.
\end{proof}

We can now define the graph $H$ used to prove the lower bound \Thm{main:planar}.
For convenience, assume that $n$ is a multiple of $t|\cF|$, the size of graphs in $\cF$.
(If not, we will pad the final graph with extra isolated vertices.) The (unlabeled)
graph $H$ is a disjoint union of $n/t|\cF|$ copies of every graph in $\cF$.

Analogously, for $R \subset I, |R| = |\cF|/2$, let $H_R$
be the graph that has $n/(2t|\cF|)$ copies of every graph in $\cF$ indexed by $R$.
For convenience, we refer to any such $R$ as a \emph{half index set}.

\subsection{The distinguishing problem} 

Note that input graph is really a labeled instance, as given by the adjacency list. 
We now give the explicit distribution of labeled inputs $G$ and the main ``distinguishing" problem.
By Yao's minimax lemma, it suffices to prove lower bounds for determinstic algorithms
over randomized inputs. So let us define the YES and NO distributions.

\begin{itemize}
    \item YES: Generate a uar labeling of $H$.
    \item NO: Generate a uar half index set $R$. Then, generate a uar labeling of $H_R$.
\end{itemize}

Given an input $G$ generated from either of these distribution, the \emph{distinguishing problem}
is to determine, with probability $> 2/3$, which distribution $G$ came from. (Note that the supports of these distributions are disjoint.)

\begin{claim} \label{clm:distinguish} If there exists a deterministic distinguisher that makes $q$ queries
(where $q$ is independent of $n$), then there exists an algorithm with the following property. Given input $G$
(generated as above), the algorithm gets $q$ uar connected components of $G$ and determines, 
with probability $> 2/3 - o(1)$, whether $G$ came from the YES or NO distribution.
\end{claim}

\begin{proof} Consider any deterministic algorithm. One can express its behavior as follows.
At any stage, it has explored the adjacency lists of some vertices $v_1, v_2, \ldots, v_b$.
It can choose to further query the adjacency lists of these vertices. In this case, it further
explores the connected components already visited. The algorithm can also choose to visit a new 
vertex ID $v'$. Note that $v'$ is an arbitrary (but fixed) function of all the vertex
IDs and edges seen so far. Since the input is a uar labeling of either $H$ or $H_R$,
the probability that $v'$ lies in a component already visited is at most $qt/n$.
With probability at least $1-qt/n$, $v'$ lies in a uar connected component (since all components
have exactly the same size, $t$). By a union bound, 
with probability at least $1-q^2t/n = 1 - o(1)$, the algorithm sees a uar set of at most $q$ components
and makes its decision.

The total number of components is $n/t$, and hence the probability of seeing a specific set of $q$
components is (up to $1\pm o(1)$-factors) $1/{n/t \choose q}$. Since $n$ is sufficiently large, this probability is within $(1+o(1))$-factors
of $q!/(n/t)^q$. This implies that with probability at least $1-o(1)$, the algorithm sees a uniform distribution
of multisets of $q$ components. Alternately, the algorithm gets $q$ uar connected components of $G$.
The success probability only changes by $(1-o(1))$-factors.
\end{proof}

The proof of \Lem{main} is now a straightforward reduction from distribution testing
and an application of the tools built thus far.

\begin{proof} (of \Lem{main}) By \Clm{main:item2}, all graphs generated by the NO distribution
are $\eps$-far from the property $\{H\}$. By Yao's minimax lemma, it suffices
to show that any deterministic distinguisher requires $\Omega(\sqrt{|\cF|})$ queries.

By \Clm{distinguish}, if there is a deterministic distinguisher making $q$ queries,
then there is an algorithm that can distinguish $H$ from $H_R$ (with probability $> 2/3 - o(1)$) 
by querying $q$ uar connected components. Standard arguments in distribution 
testing show that $\Omega(\sqrt{|\cF|})$ samples are required to distinguish the uniform distribution on the index set $I = [|\cF|]$
from the uniform distribution on a uar half index set $R$ (refer to the section titled ``An $\Omega(\sqrt{n}/\eps^2)$ lower bound" on page 13 of \cite{Clemsurvey}).
By mapping $I$ arbitrarily to $\cF$, getting samples from the uniform distribution in $I$
is equivalent to getting uar components from $G$ generated from the YES distribution.
Analogously, samples from the uniform distribution on a uar half index set $R$
is equivalent to uar components from $G$ generated from the NO distribution.
Thus, any algorithm that distinguishes $H$ from $H_R$ with uar connected components
requires $\Omega(\sqrt{|\cF|})$ samples, leading to the lower bound for deterministic
distinguishers.
\end{proof}

\section{A suitable family for planar graphs}

The main result of this section is the following theorem.

\begin{theorem} \label{thm:large:family}
	There exists a family $\cF$ of planar graphs where each graph has $s^2$ vertices and
	\begin{asparaitem}
	\item $|\cF| \geq \exp(\Omega(s^2))$.
	\item For every $G \in \cF$, the size of a minimum balanced vertex separator in $G$ is $\Omega(s)$
	\item For every pair of graphs $G,G' \in \cF$, it takes at least $0.16 s^2$ edge edits to go from
		$G$ to $G'$. In particular, it holds that $G$ and $G'$ are $0.02$-far.
	\end{asparaitem}
\end{theorem}

We collect some ingredients which will be useful in proving \Thm{large:family}. 
The first important ingredient we need is Whitney's theorem.

\begin{theorem}[Whitney's Theorem. (Theorem 4.3.2 in \cite{Diestel})]
	Any two planar embeddings of a $3$-connected planar graph are equivalent. \label{thm:whitney}
\end{theorem}

We first define a ``base'' graph (\Def{base:graph}). This is obtained in the following manner. Let us start with the $s \times s$ 
grid which we denote as $G_0$. The label set of $G_0$ is indexed by a pair 
$(i,j) \in [s] \times [s]$. We call this labeling the \emph{standard grid order}. We denote 
the base graph by $G$, and $G$ is constructed on top of $G_0$ by the addition of some specific edges. Essentially, we add edges at the corners; we refer the reader to \Fig{grid}. We add edges from $(0,0)$ to $(0,2),(2,0)$, from $(0,s)$ to $(0,s-3),(3,s)$, from $(s,s)$ to $(s,s-4),(s-4,s)$ and from $(s,0)$ to $(s,5),(s-5,0)$. This graph $G$ thus differs from  $G_0$
only in the edges adjacent to the corner vertices, making all the corners \emph{unique}. Let us label $G$ according to the standard
grid order. For $u \in V$, we thus identify $u$ as the pair $(x_u, y_u)$ which refers to the $(x,y)$ 
coordinate of $u$ on the grid. With this labeling, we have $E(G) = E(G_0) \cup E_1$.
where the set $E_1$ contains all the non grid neighbors of all the corner vertices.

\begin{definition} \label{def:base:graph}
	Consider the graph $G$ obtained above and consider the planar embedding of $G$ (which is unique by
	\Thm{whitney}). We will fix an embedding of $G$ where 

	\begin{asparaitem}
	\item The unique corner vertex adjacent to two vertices 2 hops away (according to standard grid order) 
		is located at $(0,0)$.
	\item The unique corner vertex adjacent to two vertices 3 hops away (in the standard grid order) 
		is located at $(0,s)$.
	\item The unique corner vertex adjacent to two vertices 4 hops away (in the standard grid order) 
		is located at $(s,s)$. 	
	\item The unique corner vertex adjacent to two vertices 5 hops away (in the standard grid order) 
	is located at $(s,0)$.
	\end{asparaitem}
	We call this graph (and by abuse of terminology, its embedding) as the base graph.
\end{definition}

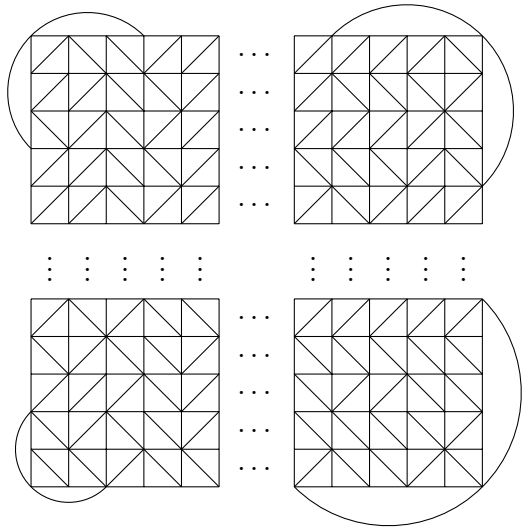
\begin{figure}[h]
	\caption{An example of one of the graphs from our family $\mathcal B$.}
	\centering 
	\begin{tikzpicture}[every node/.style={minimum size=.5cm-\pgflinewidth, outer sep=0pt}]
	\draw[step=0.5cm,color=black] (0,0) grid (2.5,2.5);
	\node at (3,1.75) {\ldots};
	\node at (3,1.25) {\ldots};
	\node at (3,0.75) {\ldots};
	\node at (3,0.25) {\ldots};
	\node at (3,2.25) {\ldots};
	\draw[step=0.5cm,color=black] (3.4999,0) grid (6,2.5);
	\node at (0.25,3) {\vdots};
	\node at (0.75,3) {\vdots};
	\node at (1.25,3) {\vdots};
	\node at (1.75,3) {\vdots};
	\node at (2.25,3) {\vdots};
	\draw[step=0.5cm,color=black] (0,3.499) grid (2.5,6);
	\node at (3.75,3) {\vdots};
	\node at (4.75,3) {\vdots};
	\node at (4.25,3) {\vdots};
	\node at (5.25,3) {\vdots};
	\node at (5.75,3) {\vdots};
	\draw[step=0.5cm,color=black] (3.499,3.499) grid (6,6);
	\node at (3,3.75) {\ldots};
	\node at (3,4.25) {\ldots};
	\node at (3,4.75) {\ldots};
	\node at (3,5.25) {\ldots};
	\node at (3,5.75) {\ldots};
	\coordinate (BL) at (0,0);
	\coordinate (BLR) at (1,0);
	\coordinate (BLU) at (0,1);
	\coordinate (TL) at (0,6);
	\coordinate (TLR) at (1.5,6);
	\coordinate (TLD) at (0,4.5);
	\coordinate (TR) at (6,6);
	\coordinate (TRL) at (4,6);
	\coordinate (TRD) at (6,4);
	\coordinate (BR) at (6,0);
	\coordinate (BRL) at (3.5,0);
	\coordinate (BRU) at (6,2.5);
	\draw (TL) to[out=45, in=135] (TLR);
	\draw (TL) to[out=225, in=135] (TLD);
	\draw (BL) to[out=315, in=225] (BLR);
	\draw (BL) to[out=135, in=225] (BLU);
	\draw (BR) to[out=225, in=315] (BRL);
	\draw (BR) to[out=45, in=315] (BRU);
	\draw (TR) to[out=135, in=45] (TRL);
	\draw (TR) to[out=315, in=45] (TRD);

	\draw ( 0.5 , 0.0 ) -- ( 0.0 , 0.5 );
	\draw ( 0.5 , 0.5 ) -- ( 0.0 , 1.0 );
	\draw ( 0.0 , 1.0 ) -- ( 0.5 , 1.5 );
	\draw ( 0.5 , 1.5 ) -- ( 0.0 , 2.0 );
	\draw ( 0.0 , 2.0 ) -- ( 0.5 , 2.5 );
	\draw ( 0.0 , 3.5 ) -- ( 0.5 , 4.0 );
	\draw ( 0.0 , 4.0 ) -- ( 0.5 , 4.5 );
	\draw ( 0.5 , 4.5 ) -- ( 0.0 , 5.0 );
	\draw ( 0.0 , 5.0 ) -- ( 0.5 , 5.5 );
	\draw ( 0.0 , 5.5 ) -- ( 0.5 , 6.0 );
	\draw ( 1.0 , 0.0 ) -- ( 0.5 , 0.5 );
	\draw ( 0.5 , 0.5 ) -- ( 1.0 , 1.0 );
	\draw ( 0.5 , 1.0 ) -- ( 1.0 , 1.5 );
	\draw ( 0.5 , 1.5 ) -- ( 1.0 , 2.0 );
	\draw ( 1.0 , 2.0 ) -- ( 0.5 , 2.5 );
	\draw ( 0.5 , 3.5 ) -- ( 1.0 , 4.0 );
	\draw ( 0.5 , 4.0 ) -- ( 1.0 , 4.5 );
	\draw ( 0.5 , 4.5 ) -- ( 1.0 , 5.0 );
	\draw ( 0.5 , 5.0 ) -- ( 1.0 , 5.5 );
	\draw ( 1.0 , 5.5 ) -- ( 0.5 , 6.0 );
	\draw ( 1.0 , 0.0 ) -- ( 1.5 , 0.5 );
	\draw ( 1.5 , 0.5 ) -- ( 1.0 , 1.0 );
	\draw ( 1.0 , 1.0 ) -- ( 1.5 , 1.5 );
	\draw ( 1.5 , 1.5 ) -- ( 1.0 , 2.0 );
	\draw ( 1.0 , 2.0 ) -- ( 1.5 , 2.5 );
	\draw ( 1.5 , 3.5 ) -- ( 1.0 , 4.0 );
	\draw ( 1.5 , 4.0 ) -- ( 1.0 , 4.5 );
	\draw ( 1.5 , 4.5 ) -- ( 1.0 , 5.0 );
	\draw ( 1.5 , 5.0 ) -- ( 1.0 , 5.5 );
	\draw ( 1.5 , 5.5 ) -- ( 1.0 , 6.0 );
	\draw ( 2.0 , 0.0 ) -- ( 1.5 , 0.5 );
	\draw ( 2.0 , 0.5 ) -- ( 1.5 , 1.0 );
	\draw ( 2.0 , 1.0 ) -- ( 1.5 , 1.5 );
	\draw ( 2.0 , 1.5 ) -- ( 1.5 , 2.0 );
	\draw ( 2.0 , 2.0 ) -- ( 1.5 , 2.5 );
	\draw ( 1.5 , 3.5 ) -- ( 2.0 , 4.0 );
	\draw ( 2.0 , 4.0 ) -- ( 1.5 , 4.5 );
	\draw ( 1.5 , 4.5 ) -- ( 2.0 , 5.0 );
	\draw ( 2.0 , 5.0 ) -- ( 1.5 , 5.5 );
	\draw ( 1.5 , 5.5 ) -- ( 2.0 , 6.0 );
	\draw ( 2.5 , 0.0 ) -- ( 2.0 , 0.5 );
	\draw ( 2.0 , 0.5 ) -- ( 2.5 , 1.0 );
	\draw ( 2.0 , 1.0 ) -- ( 2.5 , 1.5 );
	\draw ( 2.0 , 1.5 ) -- ( 2.5 , 2.0 );
	\draw ( 2.5 , 2.0 ) -- ( 2.0 , 2.5 );
	\draw ( 2.0 , 3.5 ) -- ( 2.5 , 4.0 );
	\draw ( 2.0 , 4.0 ) -- ( 2.5 , 4.5 );
	\draw ( 2.0 , 4.5 ) -- ( 2.5 , 5.0 );
	\draw ( 2.0 , 5.0 ) -- ( 2.5 , 5.5 );
	\draw ( 2.0 , 5.5 ) -- ( 2.5 , 6.0 );
	\draw ( 3.5 , 0.0 ) -- ( 4.0 , 0.5 );
	\draw ( 4.0 , 0.5 ) -- ( 3.5 , 1.0 );
	\draw ( 4.0 , 1.0 ) -- ( 3.5 , 1.5 );
	\draw ( 4.0 , 1.5 ) -- ( 3.5 , 2.0 );
	\draw ( 3.5 , 2.0 ) -- ( 4.0 , 2.5 );
	\draw ( 3.5 , 3.5 ) -- ( 4.0 , 4.0 );
	\draw ( 4.0 , 4.0 ) -- ( 3.5 , 4.5 );
	\draw ( 3.5 , 4.5 ) -- ( 4.0 , 5.0 );
	\draw ( 3.5 , 5.0 ) -- ( 4.0 , 5.5 );
	\draw ( 3.5 , 5.5 ) -- ( 4.0 , 6.0 );
	\draw ( 4.5 , 0.0 ) -- ( 4.0 , 0.5 );
	\draw ( 4.5 , 0.5 ) -- ( 4.0 , 1.0 );
	\draw ( 4.0 , 1.0 ) -- ( 4.5 , 1.5 );
	\draw ( 4.5 , 1.5 ) -- ( 4.0 , 2.0 );
	\draw ( 4.0 , 2.0 ) -- ( 4.5 , 2.5 );
	\draw ( 4.5 , 3.5 ) -- ( 4.0 , 4.0 );
	\draw ( 4.5 , 4.0 ) -- ( 4.0 , 4.5 );
	\draw ( 4.0 , 4.5 ) -- ( 4.5 , 5.0 );
	\draw ( 4.5 , 5.0 ) -- ( 4.0 , 5.5 );
	\draw ( 4.0 , 5.5 ) -- ( 4.5 , 6.0 );
	\draw ( 5.0 , 0.0 ) -- ( 4.5 , 0.5 );
	\draw ( 5.0 , 0.5 ) -- ( 4.5 , 1.0 );
	\draw ( 4.5 , 1.0 ) -- ( 5.0 , 1.5 );
	\draw ( 5.0 , 1.5 ) -- ( 4.5 , 2.0 );
	\draw ( 4.5 , 2.0 ) -- ( 5.0 , 2.5 );
	\draw ( 4.5 , 3.5 ) -- ( 5.0 , 4.0 );
	\draw ( 5.0 , 4.0 ) -- ( 4.5 , 4.5 );
	\draw ( 4.5 , 4.5 ) -- ( 5.0 , 5.0 );
	\draw ( 5.0 , 5.0 ) -- ( 4.5 , 5.5 );
	\draw ( 4.5 , 5.5 ) -- ( 5.0 , 6.0 );
	\draw ( 5.0 , 0.0 ) -- ( 5.5 , 0.5 );
	\draw ( 5.5 , 0.5 ) -- ( 5.0 , 1.0 );
	\draw ( 5.5 , 1.0 ) -- ( 5.0 , 1.5 );
	\draw ( 5.5 , 1.5 ) -- ( 5.0 , 2.0 );
	\draw ( 5.0 , 2.0 ) -- ( 5.5 , 2.5 );
	\draw ( 5.0 , 3.5 ) -- ( 5.5 , 4.0 );
	\draw ( 5.0 , 4.0 ) -- ( 5.5 , 4.5 );
	\draw ( 5.0 , 4.5 ) -- ( 5.5 , 5.0 );
	\draw ( 5.5 , 5.0 ) -- ( 5.0 , 5.5 );
	\draw ( 5.0 , 5.5 ) -- ( 5.5 , 6.0 );
	\draw ( 6.0 , 0.0 ) -- ( 5.5 , 0.5 );
	\draw ( 5.5 , 0.5 ) -- ( 6.0 , 1.0 );
	\draw ( 6.0 , 1.0 ) -- ( 5.5 , 1.5 );
	\draw ( 6.0 , 1.5 ) -- ( 5.5 , 2.0 );
	\draw ( 5.5 , 2.0 ) -- ( 6.0 , 2.5 );
	\draw ( 5.5 , 3.5 ) -- ( 6.0 , 4.0 );
	\draw ( 6.0 , 4.0 ) -- ( 5.5 , 4.5 );
	\draw ( 6.0 , 4.5 ) -- ( 5.5 , 5.0 );
	\draw ( 5.5 , 5.0 ) -- ( 6.0 , 5.5 );
	\draw ( 6.0 , 5.5 ) -- ( 5.5 , 6.0 );
	\end{tikzpicture}\label{fig:grid}
\end{figure}

\subsection{Proof of \Thm{large:family}}

Let us begin by defining the following family of graphs. We will show that all graphs in this family
are pairwise non isomorphic.

\begin{construction} \label{con:large:family}
	The family $\cB$ of graphs is a family of graphs labeled according to \emph{standard grid order} 
	which is obtained from the base graph $G$ in the following way. Let 
	$$\cI =  [0, s) \times (1, s]$$ 
	denote an index set for vertices in $V(G)$.
	Graphs in this family are obtained in the following manner. For each $(i,j) \in \cI$, we add
	exactly one of the following edges. Either we add
	\begin{itemize}
		\item The upward diagonal $\left((i,j), (i+1, j+1) \right)$, or
		\item The downward diagonal $\left( (i+1, j), (i+1, j-1) \right)$.
	\end{itemize}

	Noting $|\cI| = s^2$, note that the size of this family is $|\cB| = 2^{s^2}$.
\end{construction}

\begin{lemma}
	The family $\cB$ defined in \Con{large:family} is a collection of $2^{s^2}$ pairwise 
	non isomorphic graphs on $s^2$ vertices. 
\end{lemma}

\begin{proof}
	As mentioned earlier in \Def{base:graph}, the base graph $G$ admits
	a unique planar drawing. This is due to \Thm{whitney}, because the additional edges made our graph $3$-connected. We fix that embedding. Now consider two 
	drawings $B_1, B_2 \in \cB$. They differ in at least one diagonal and 
	therefore they are non isomorphic.
\end{proof}

However, we still need to show that this family has a sufficient number of graphs that are far 
from being isomorphic to each other. We state this in the following lemma.

\begin{lemma}
	There is a greedy procedure that takes as input the family $\cB$ of graphs and returns 
	a family $\cF$ of size at least $2^{s^2/5}$ all of which are pairwise $0.01$-far
	from each other. \label{lem:prop:lf}
\end{lemma}

\begin{proof}
	We begin with the following observations. Consider the 
	ball $B_b(G)$ for any instance of the modified grid $G$, which is the set of graphs reachable 
	from $G$ by $b$ (diagonal) edge deletions. For any $G \in \cB$, 
	note that the number of edges in $G$ is at most $s^2$. Thus, the number of possible graphs reachable 
	from $G$ after making $b$ deletions is at most ${s^2 \choose b}$. We replace these deleted diagonals
	with diagonals oriented the other way. Choose $b = 0.16 s^2$. We apply the bound
    ${a \choose b} \leq (ea/b)^b$.
	\begin{align*}
		\max_{G \in \cB} |B_b(G)| &\leq {s^2\choose b} \leq \left(\frac{es^2}{b}\right)^b \leq \left(\frac{es^2}{0.16 s^2}\right)^{0.16s^2} \leq 2^{0.8 s^2}
	\end{align*}

	Let $\cF$ denote the set found by the algorithm below. Using the above bound on $\max_{G \in \cB} |B_b(G)|$,
	it holds that 
	$$|\cF| \geq \frac{2^{s^2}}{\max_{G \in \cB} \cB_{0.16 s^2}(G)} \geq \frac{2^{s^2}}{2^{0.8s^2}} \geq 2^{s^2/5}.$$
	Since for any $G \in \cB$, the maximum degree of any vertex in $G$ is at most $8$,
	and all the graphs in the set $\cF$ disagree in at least $0.16 s^2$ edges, it follows that all the graphs
	in $\cF$ are pairwise $0.02$-far. Below, we present our greedy procedure.

\hspace{1cm}

%

\noindent
\fbox{
	\begin{minipage}{0.9\textwidth}
		{TakeScoops($\cB$)}
		\smallskip
		\begin{compactenum}
			\item Initialize $\cF =\emptyset$.	
			\item \textbf{While} {$\cB \neq \emptyset$}:
			\begin{compactenum}
				\item $\cF =\emptyset$.
				\item $\cB = \cB \setminus B_{0.16s^2}(G)$.
				\item $\cF = \cF \cup G$.
			\end{compactenum}
			\textbf{end while}
			\item \textbf{return} $\cF$.
		\end{compactenum}
\end{minipage}} \\

\hspace{1cm}
	The family $\cF$ is thus the family output with the desired properties as stated in 
	\Thm{large:family}.

\end{proof}

The proof of the theorem follows immediately from the lemmas.
\begin{proof}[Proof of \Thm{large:family}]
	We show that the family $\cF$ obtained in \Lem{prop:lf}
	satisfies all the criteria we require to be satisfied in \Thm{large:family}. We pick these one by one:
	\begin{itemize}
		\item $\underline{ |\cF|\geq\exp(\Omega(s^2)) :}$ This is a direct consequence of \Lem{prop:lf}.
		\item $\underline{ \textrm{For every } G \in \cF, \textrm{ the size of the minimum balanced separator is } 
			\Omega(s): }$ This follows from the fact that the family $\cF$ comprises graphs that 
			contain the grid on $s\times s$ vertices as proper induced subgraphs. The grid has 
			balanced separators of size $\Omega(s)$, so all graphs in $\cF$ must have separators 
			at least as large. 
		\item $\underline{\textrm{Graphs in }\cF \textrm{ are all pairwise } 0.01 \textrm{ far from each other:}}$ 
			Follows from \Lem{prop:lf}. 
	\end{itemize}
	This shows that $\mathcal{F}$ follows the properties described in theorem \ref{thm:large:family}, thus completing its proof.
\end{proof}

\subsubsection{Proof of \Thm{main:planar}} \label{sec:proof-lb-main}

The proof follows from plugging in the right value for $s$. In \Thm{large:family}, we set $s = \eps^{-1}$.
We get a suitable family $\cF$ with $t = \eps^{-2}$, where $|\cF| = \exp(\Omega(\eps^{-2}))$.
By \Lem{main}, we have a explicit singleton property that requires $\Omega(\sqrt{|\cF|}) = \exp(\Omega(\eps^{-2}))$
queries to test.

\section{A suitable family of trees} \label{sec:suitable-tree}

This construction is fairly straightforward: we simply pick all distinct
trees on $\eps^{-1}$ vertices. We need to take care to count the number
of \emph{unrooted} trees, so some work is required.

Let $\unlabr(s)$ denote the set of unlabled and rooted binary trees on $s$ vertices. 
We first recall the following fact.

\begin{fact} \label{fact:tree:cnt}
	The number of different unlabeled rooted binary trees on $s$ vertices is 
	$|\unlabr(s)| \geq 2^{\Omega(s)}$ (Probelm 1.44 in \cite{FSBook} where they 
	attribute this result to \cite{Otter}).
\end{fact}

We would like to lower bound the number of unlabeled, unrooted and binary trees on $s$ vertices. To this 
end, we first make the following definition.

\begin{definition} \label{def:orbit}
	Write $\unlabr(s) = \{(T,r) \colon T \text{ is unlableled with root } r\}$. 
	For $(T,r) \in \unlabr(n)$, let 
	$$Orbit(T) = \{(T',r') \in \unlabr(n): T' \cong T \}.$$
	Note that this also means $T \in Orbit(T')$. Also, note that if $(T'',r'') \not\in Orbit(T)$
	then $Orbit(T'') \cap Orbit(T) = \emptyset$.
\end{definition}

Now, we note the following.

\begin{claim}
	For any $(T,r) \in \unlabr(s)$, $|Orbit(T)| \leq s$.
\end{claim}

\begin{proof}
	Suppose not. Then $\exists (T', x), (T'', x) \in Orbit(T)$ which share $x$ as the root
	where both $T \cong T'$ and $T \cong T''$. Then $T' \cong T''$ as well. However,
	we also have $(T', x), (T'',x) \in \unlabr(n)$. Contradiction.
\end{proof}

Now, we will produce a family of unlabeled and unrooted binary trees. This is done by picking
a single tree from each orbit in $\unlabr(s)$. The following is immediate. 

\begin{claim} \label{clm:unlabu:size}
	$|\unlabu(s)| \geq 2^{\Omega(s)}$.
\end{claim}

\begin{proof}
	Since all orbits have size at most $s$, it holds that $|\unlabu(s)| \geq |\unlabr(s)| \geq 2^{\Omega(s)}$.
\end{proof}

We can now complete the proof of the main lower bound for forests.

\begin{proof} (of \Thm{main:tw:lb}) The suitable family is simply $\unlabu(\eps^{-1})$.
Since all graphs are of size $t= 1/\eps$, it only suffices to verify the separator bound for suitable families.
But that trivially holds, since the separator size is non-zero, which is $\Omega(1) = \Omega(\eps t)$.
The size of the family is $\exp(\eps^{-1})$, by \Clm{unlabu:size}. \Lem{main} gives the lower
bound of $\exp(\eps^{-1})$ for property testing an explicit singleton property.
\end{proof}

\bibliographystyle{alpha}
\bibliography{ref}

\begin{appendix}
\section{Missing proofs from \Sec{up}}


In this appendix, our objective is to prove \Clm{count:tw} and \Clm{po:on:tw}. The former claim counts the number
of graphs with treewidth at most $\tau$ and is proved below.

\begin{claim} \label{clm:count:tw}
	The number of unlabeled bounded treewidth graphs on $t$ vertices is at most $\exp{O(t\log t)}$.
\end{claim}

\begin{proof}
 	Observe that for a graph on $t$ vertices with bounded treewidth $\tau=O(1)$ has $O(t)$ edges. Then, the number of unlabelled graphs of this form, $T_t$ is at most ${N\choose O(t)}$, where $N={t\choose 2}$. Thus $T_t=t^O(t)$, which implies that the number of bounded treewidth graphs on $t$ vertices is at most $\exp(O(t\log t))$.
\end{proof}

Next, we produce the statment of \Clm{po:on:tw}. It postulates that there exist efficient 
$poly(d \eps^{-1})$ time implementations of partition oracles for $d$-degree bounde graphs 
with bounded treewidth. 

One key tool we need to prove this theorem is an analog of Alon-Seymour-Thomas style result which
proves that a graph with treewidth at most $\tau$ is $(\eps, O(\tau/\eps))$-hyperfinite. 
%

Let $\tilde{s}(G)$ be the size of the smallest $1/3-2/3$ balanced separator, 
whose removal leads to connected components of size at most $2/3$ of the vertices.

\begin{theorem} \label{thm:treewidth} 
[Theorem 2 of \cite{Gruber13} and Theorem 12 of \cite{BODLAENDER1995238}]
$\tilde{s}(G)-1\leq tw(G)$.
\end{theorem}

We now prove that every treewidth $\tau$ graph is $(\eps, 6\tau/\eps)$ hyperfinite. 

\begin{lemma} \label{lem:ASTtype}
	Fix $\eps > 0$, $d, \tau \in \N$. Then for any graph $G$ with $tw(G) \leq \tau-1$ 
	there exists a set of edges $E'$ such that the following hold:
	\begin{itemize}
		\item $G\setminus E'$ is a graph with connected components no larger than $6\tau/\eps$. 
		\item $|E'| \leq \eps dn$
	\end{itemize}
\end{lemma}

\begin{proof}
	We begin by showing item 1. Since $G$ has treewidth at most $\tau$, all of its
	subgraphs also have treewidth $\tau$. Also, recall $tw(G) \leq \tau$, implies
	that for all $1/3$-balanced separators of $G$ satisfy $|\balsep(G)| \leq O(\tau)$. 

	Consider the following recursive procedure. Here $\delsep(G)$ is a routine which returns
	a three tuple $(S, G_0, G_1)$ where $S$ is the smallest balanced separator of $G$ and 
	$G_0$ and $G_1$ are collections of connected components of $G$ each with size at least
	$|V(G)|/3$. \\

	\noindent
	\fbox{
		\begin{minipage}{0.9\textwidth}
			{\decompose$(G)$}
			\smallskip
			\begin{compactenum}
			\item If $|V(G)| \geq 6\tau/\eps$
				\begin{compactenum}
				\item Write $(S_0, G_0, G_1) = \delsep(G)$.
				\item $\decompose(G_0), \decompose(G_1)$.
				\end{compactenum}
			\end{compactenum}
	\end{minipage}} \\

%
	Note that the sets $G_0$ and $G_1$ produced at any intermediate step of this
	procedure are \emph{unions of connected components} (and they are not necessarily
	connected components themselves). Consider the recursion tree of this process and
	let us examine the leaves of this recursion tree. The leaf nodes correspond to 
	some subgraph $G$ (which is a union of connected components) obtained after repeated 
	applications of balanced separator routine such that $|V(G)| \leq 6 \tau/\eps$.
	Also, since the procedure repeatedly deletes a $1/3$-separator, it also holds that
	all the leaf nodes of the recursion tree correspond to graphs with at least 
	$2\tau/\eps$ vertices. Thus, since the leaves correspond to disjoint subgraphs, the number of leaves equals $L \leq \frac{\eps n}{2\tau}$
	and the total number of nodes produced in the recursion is at most $2L$. Generating
	every node requires deleting at most $\tau$ edges. Thus, in all, we delete
	$\eps dn$ edges which proves the lemma.


\end{proof}

%
%
%
%

We now prove bounds on partition oracles for bounded treewidth graphs.

\begin{restatable}{claim}{poontw} \label{clm:po:on:tw}
		Let $\cT_{\tau}$ denote the class of graphs with treewidth at most $\tau$. 
		Let $G \in \cT_{\tau}$ denote a graph with treewidth at most $\tau$. Then $G$ admits 
		an $(\eps, k)$ partition oracle with $k \leq \frac{30 \tau}{\eps}$.
\end{restatable}

\begin{proof} 
	Set $\alpha = \eps/2$. Since $\cT_{\tau}$ is a minor-closed class of graphs, it follows from \Thm{kss:21}
	that any $G \in \cT_{\tau}$ admits an $(\alpha, O(\alpha^{-2}))$ partition oracle whose setting on $G$
	is denoted as $\cP(G)$. Note that the number
	of edges that run between the connected components of the underlying partition is at most 
	$\alpha nd \leq \half \cdot \eps dn$.
	We will show how to modify the parition obtained by this oracle and obtain another oracle which returns
	components with size at most $30 \tau/\eps$. This is done as follows. 

	Let $P \in \cP(G)$ denote some connected component in the partition $\cP(G)$. If $|P| \leq 30 \tau/\eps$, 
	we do not refine $P$. On the other hand, if $|P| > 30 \tau/\eps$, we use \Lem{ASTtype} on $P$
	with parameter $\beta = \eps/2$. This refines $P$ and produces connected components of size
	at most $3 \tau/\beta \leq 6 \tau/\eps \leq 30 \tau/\eps$ as desired. And the total number of edges lost is
	at most $\beta dn \leq \half \cdot \eps dn$. Overall, the total number of edges lost is 
	at most $\eps dn$ as desired.
\end{proof}

\end{appendix}

\end{document}